\documentclass[preprint]{acm_proc_article-sp}

\usepackage{dsfont}

\usepackage[pdftex,colorlinks,%
           bookmarks=true,%
           pdftitle=Pioneers of Influence Propagation in Social Networks,%
           pdfauthor=K.\ Gaurav\ -B.\ Blaszczyszyn%
          \ -\  \ P.\ Keeler]{hyperref}
\hypersetup{
   bookmarksnumbered,
   pdfstartview={FitH},
   citecolor={blue},
   linkcolor={red},
   urlcolor=[rgb]{0,0.55,0},
   pdfpagemode={UseOutlines}
}

\newtheorem{thm}{Theorem}[section]

\newtheorem{remark}[thm]{Remark}

\newtheorem{defin}[thm]{Definition}

\newtheorem{example}[thm]{ Example}
\newenvironment{exe}[1][]{\begin{example}[#1]\em}{\end{example}}
\newtheorem{fact}[thm]{Fact}
\newtheorem{claim}[thm]{Claim}

\renewcommand{\P}{\mathbb{P}}
\newcommand{\E}{\mathbb{E}}

\begin{document}
%
\conferenceinfo{}{}
\CopyrightYear{} 
\crdata{}  
\CopyrightYear{}

\title{Pioneers of Influence Propagation in Social Networks}
%
%
%
%
%

\numberofauthors{3} 
%
\author{
%
%
\alignauthor
Kumar Gaurav
\\
       \affaddr{UPMC/Inria/ENS}\\
       \affaddr{23 avenue d'Italie}\\
       \affaddr{Paris, France}\\
       \email{Kumar.Gaurav@inria.fr}
\alignauthor
Bart{\l}omiej~B{\l}aszczyszyn 
\\
       \affaddr{Inria/ENS}\\
       \affaddr{23 avenue d'Italie}\\
       \affaddr{Paris, France}\\
       \email{Bartek.Blaszczyszyn@ens.fr}
\alignauthor
Paul Holger Keeler
\\
       \affaddr{Inria/ENS}\\
       \affaddr{23 avenue d'Italie}\\
       \affaddr{Paris, France}\\
       \email{Paul.Keeler@inria.fr}
}




\maketitle
\begin{abstract}
With the growing importance of corporate viral marketing campaigns on online social networks,
the interest in studies of influence propagation through networks is higher than ever. In a
viral marketing campaign, a firm initially targets a small set of pioneers and hopes that
they would influence a sizeable
fraction of the
population by diffusion of influence through the network. In general, any marketing campaign might fail to go viral in
the first try. As such, it would be useful to have some guide to evaluate 
the effectiveness of the campaign and judge whether it is worthy of further resources, and
in case the campaign has potential, how to hit upon
a good pioneer who can make the campaign go viral.

In this paper, we present a diffusion model developed by enriching the generalized random graph (a.k.a.
configuration model) to provide insight into these questions. We offer the intuition behind the results
on this model, rigorously proved in \cite{viral}, and illustrate them here by taking examples of random networks 
having prototypical degree distributions --- Poisson degree distribution, which is commonly used as a kind of benchmark, and
Power Law degree distribution, which is normally used to approximate the real-world networks. On these networks, the members
are assumed to have varying attitudes towards propagating the information. We analyze three cases, in particular --- (1)~Bernoulli
transmissions, when a member influences each of its friend with
probability $p$; (2)~Node percolation, when a member influences
all its friends with probability $p$ and none with probability $1-p$;
(3)~Coupon-collector transmissions, when a member randomly
selects one of his friends $K$ times with replacement.

We assume that the configuration model is the closest approximation of a large online social network, when the information
available about the network is very limited.
The key insight offered by this study from a firm's perspective is regarding how to evaluate the effectiveness of a 
marketing campaign and do cost-benefit analysis by collecting relevant statistical data from the pioneers it selects.
The campaign evaluation criterion is informed by the observation that if the parameters of the underlying network and 
the campaign effectiveness are such that the campaign
can indeed reach a significant fraction of the population, then the set of good pioneers also
forms a significant fraction of the population. Therefore, in such a case, the firms can even adopt the na\"{i}ve strategy of 
repeatedly picking and targeting some number of pioneers at random from the population. With this strategy, the probability of them 
picking a good pioneer
will increase geometrically fast with the number of tries.

\end{abstract}




\section{Introduction}
Traditionally, firms had few avenues when trying to market their products. And the most important of these
avenues --- television, newspapers, billboards --- are notoriously inflexible and inefficient from the firms' point
of view. Essentially, a firm has to pay to reach even those people who would never form a part of its target
demographic (\cite{maxi}). From the consumers' point of view, they are continuously bombarded with advertisements of products,
a vast majority of which do not interest them. In such a scenario, there is even a possibility that a significant
fraction of consumers might just \textit{tune-off} and become insensitive to every advertisement. The idea of \textit{direct marketing}
tried to overcome some of these problems by the construction of a database of the buying patterns and other relevent 
information of the population, and then targeting only those who are predisposed to get influenced by a particular 
marketing campaign (\cite{And98datamining}).
However, targeting the most responsive customers individually can be expensive and thus limits the reach of 
direct marketing. Moreover, it precludes the possibility of positive externalities such as a favorable shift
in preferences of a demographic segment previously thought to be unresponsive.

The penetration of internet and the emergence of huge online social networks in the last decade has radically
altered the way that people consume media and print, leading to an ongoing decline in importance of conventional
channels and consequently, marketing through them. This radical shift has brought in its wake a host of 
opportunities as well as challenges for the advertisers. On the one hand, firms finally have the possibility
to reach in a cost-effective way not only the past responsive customers, but indeed all the potentially responsive
ones. The importance of this new marketing medium is witnessed by the fact that most of the big corporations, particularly
those providing services or producing consumer goods, now have dedicated \textit{fan-pages} on social networks
to interact with their loyal customers. These, in turn, help the firms spread their new marketing campaign to a
large fraction of the network, reliably and at a fraction of the cost incurred through traditional channels.
On the other hand, new firms without a loyal fan-base have found it a hit-or-miss game to gain attention
through the new medium. Even though the marketing through network is mostly a miss for these firms, but when
it is a hit, it is a spectacular one. This makes it tempting for firms to keep waiting for that spectacular hit
while their marketing budget inflates beyond the point of no return.
The \textit{fat-tail} uncertainty of viral marketing makes it inherently
different from conventional marketing and calls for a fundamentally different approach to
decision-making: which individuals, and how many, to initially target in the online network? What amount of 
resources to spend on these initially targeted \textit{pioneers}? And most importantly, when to stop, admit the 
inefficacy of the current campaign and develop a new one?

\subsection{Results}

In this paper, we introduce a generalized diffusion dynamic on configuration model, which serves as
a very useful approximation of an online social network, particularly when one does not have an access to a
detailed information about the network structure. The diffusion dynamic that we study on this underlying
random network is essentially this: any individual in the network influences a random subset of its neighbours,
the distribution of which depends on the effectiveness of the marketing campaign.

We illustrate large-network-limit results on this model, rigorously proved in \cite{viral}. The empirical distribution of the number of friends
that a person influences in the course of a marketing campaign is taken as a measure of the effectiveness of the campaign.
We present a condition depending
on network degree-distribution (the emperical distribution of the number of friends of a network member)
and the effectiveness of a marketing campaign which, if satisfied, will allow, with a non-negligible probability,
the campaign to go viral when started from a randomly chosen individual in the network. Given this condition,
we present an estimate of the fraction of the population that is reached when the campaign does go viral. We then show that
under the same condition, the fraction of good pioneers in the network, i.e., the individuals who if targeted
initially will lead the campaign to go viral, is non-negligible as well, and we give an estimate of this
fraction. We analyze in detail the process of influence propagation on configuration model having two types of degree-distribution:
Poisson and Power Law. Three examples illustrating the dynamic of influence propagation on these two networks are considered: (1)~Bernoulli
transmissions; (2)~Node percolation; (3)~Coupon-collector transmissions.

Based on the above analysis, we offer a practical decision-making guide for marketing on online
networks which we think would be particularly useful to firms with no prior access to detailed network
structure. Specifically, we consider the na\"{i}ve strategy of picking some number of pioneers at random from
the population, spending some fixed amount of resources on each of them and waiting to see if the campaign
goes viral, picking another batch if it does not. For this strategy, we suggest what statistical data the firm
should collect from its pioneers, and based on these, how to estimate the effectiveness of the campaign 
and make a cost-benefit analysis.

\subsection{Related Work}

While the public imagination is captured by a new viral video of a relatively unknown artist,
researchers have been trying to understand this phenomenon much before the emergence of online social networks. 
It was first studied in the context of the spread of epidemics on complex networks, whence the term 
\textit{viral} marketing originates (\cite{bailey1975mathematical},
\cite{PhysRevE.66.016128}). The impact of social network structure on the propagation of social and economic behavior
has also been recognized (\cite{RePEc:tpr:qjecon:v:107:y:1992:i:3:p:797-817}, \cite{Bikh}) and there is 
growing evidence of its importance (\cite{Banerjee26072013}).

In the context of viral marketing, broadly speaking, two approaches have developed in trying to exploit the network 
structure to maximize the probability
of marketing campaign going viral for each dollar spent. The first approach tries to locate the most influential
individuals in the network who can then be targeted to \textit{seed} the campaign (\cite{Kempe:2003:MSI:956750.956769}). 
This idea has been developed into a machine learning approach which relies on the availability of 
large databases containing detailed information regarding the network structure and the past instances of influence 
propagation to come up
with the best predictor of the most influential individuals who should be targeted for future campaigns (\cite{Domingos02miningthe}). 
Our approach, although fundamentally based  on the analysis of the most influential network members, whom we call
pioneers, differs  in its philosophy of how to apply this to make marketing decisions. 
We do not rely on locating the pioneers by data-mining the
network since the tastes of online network members shift at a rapid rate and the past can be an unreliable predictor
for the current campaign. Moreover, the network database is not necessarily accessible to every firm.
Therefore, we favor a strategy which enables
one to gain exposure to positive fat-tail events while covering his/her back. However, since we suggest a way to measure
the current campaign's effectiveness based on its ongoing diffusion in the network, it can be used to develop
better predictors even when the network information is freely accessible.

The second approach that has become popular in this context does not focus on locating
influential network members but instead on giving incentives to members to act as a conduit
for the diffusion of the campaign (\cite{Arthur:2009:PSV:1697167.1697180}). Various mechanisms 
for determining the optimal incentives have been proposed
and analyzed on random networks (\cite{Marcviral})  as well as on a deterministic network (\cite{citeulike:7102381}). 
This approach can be particularly effective for web-based service providers, e.g., 
movie-renting business, where the non-monetary incentive of using the service freely or cheaply for some period of
time can motivate people to proactively advertise to their friends. However, it is not always possible
to come up with non-monetary expensive while offering monetary incentives is not cost-effective. 
In such cases, our approach can offer a more
cost-effective alternative by leveraging the inherent tendency of a social network to percolate information without
external incentives.

A variety of marketing strategies have been conceived combining the two broad approaches that we described above. 
We hope that our approach would enrich the spectrum and further help in understanding
and exploiting the phenomenon of viral marketing.

\section{Model and thoretical claims}
\label{s.Model}
In this section, we introduce our model and informally describe the results which are rigorously proved in~\cite{viral}.

\subsection{Model}
Consider that the only information available to you about an online social network is the number of friends
that a subset of network members have, a realistic assumption if you are dealing with the biggest and the most
important social networks out there. In such a case, the best you can do is to work with a uniform random network
which agrees with the statistics that you can obtain from the available information. Such a uniform random network is obtained by
constructing what is known as {\em configuration model} (CM); cf~\cite{Remco}. This random network is realized by attaching half-edges 
to each vertex corresponding to its degree (which represents here, the number of friends) and then uniformly
pair-wise matching them to create edges. 
We assume this model of the  social network  throughout the paper and will use  interchangeably the terms  ``social
graph'' and  ``random network'' meaning precisely the CM.
 We call the vertices  of this graph ``nodes'' or ``users''
and graph neighbours ``friends''.

We consider a marketing campaign started from some initial target
called {\em pioneer} in this network. The most natural propagation dynamic to assume in the absence of any other information is that
a person influences a random subset of its friends who further propagate the campaign in the same manner.
The number of friends that a person influences depends on a particular campaign. To model this dynamic,
we enhance the configuration model by partitioning the half-edges into {\em transmitter} half-edges, those through
which the influence can flow and {\em receiver} half-edges which can only receive influence. So, if a person~A influences
his friend B in the network, then in our representation, A has a transmitter half-edge matched to the transmitter
or receiver half-edge of~B.

Let $D$ and $D^{(t)}$ denote the empirically observed distributions of total degree and transmitter degree
respectively. Empirical receiver degree distribution, $D^{(r)}$, is therefore $D - D^{(t)}$. Then we have the
following large-network-limit results, rigorously proved in~\cite{viral}, but
only informally stated here.

\subsection{Theoretical claims}
\label{ss.Claims}
\bigskip
\begin{claim}
\label{cl1}
Starting from a randomly selected pioneer, the campaign can go viral, i.e., reach a strictly positive fraction of the population, with a strictly positive probability if and only if
\begin{equation}\label{e.VA}
\mathbb{E}[D^{(t)}D] > \mathbb{E}[D^{(t)}+D].
\end{equation}

\end{claim}
Note that $\mathbb{E}[D^{(t)}D] > \mathbb{E}[D^{(t)}+D]$ implies 
\begin{equation}\label{e.CA}
\mathbb{E}[D(D-2)]>0
\end{equation}
and recall that 
this latter condition is necessary and sufficient~\footnote{
under a few additional technical assumptions, as $0<\E[D]<\infty$,
$\P\{\,D=1\}>0$, which we tacitly assume throughout the paper} for the existence of a (unique) connected component of
the underlying social graph,  called {\em big component},
encompassing a 
strictly positive fraction of its population; cf~\cite{JanLuc}.
Obviously, our campaign can go viral only within this big component.

Call {\em good pioneers} the pioneers from which the 
 campaign can go viral.
\begin{claim} \label{c1} If~(\ref{e.VA}) is satisfied then  
the population reached is, more or less, the same irrespective of 
the good pioneer chosen initially.
\end{claim}

Let $C^*$ denote the population reached by the campaign when started from a good pioneer and $\overline{C}^*$ the set
of good pioneers. 

\begin{claim} If~(\ref{e.VA}) is satisfied then  the set of
good pioneers  $\overline{C}^*$ also forms a strictly positive fraction of the population.
\end{claim}

The next claim gives the estimates on the size of $C^*$ and $\overline{C}^*$. Let

\begin{equation}
\label{Hx}
 H(x):=\mathbb{E}[D] x^2-\mathbb{E}[D^{(r)}] x-\mathbb{E}[D^{(t)}x^D]
\end{equation}
and 
\begin{equation}
\label{Hbarx}
 \overline{H}(x):=\mathbb{E}[D] x^2-\mathbb{E}[D^{(t)}x^{D^{(t)}}] - \mathbb{E}[D^{(r)}x^{D^{(t)}}]x.
\end{equation}

If condition~(\ref{e.VA}) is satisfied then  $H(x)$ and
$\overline{H}(x)$ have unique zeros in $(0,1)$. Call them $\xi$ and $\overline{\xi}$ respectively. Denote also by $G_D(x)=\E[x^D]$ and 
$G_{D^{(t)}}(x)=\E[x^{D^{(t)}}]$ the probability generating function (pgf) of $D$ and $D^{(t)}$, respectively.

\begin{claim} 
\label{c.alpha}
If~(\ref{e.VA}) is satisfied and   
$n$ denotes the size of network population, then for $n$ large,
\begin{equation}
\label{e.alpha}
 \frac{\left|C^*\right|}{n} \approx 1-G_D(\xi) =: \alpha > 0
\end{equation}
and 
\begin{equation}
\label{e.alphabar}
 \frac{\left|\overline{C}^*\right|}{n} \approx 1-G_{D^{(t)}}(\bar\xi) =: \overline\alpha > 0.
\end{equation}
\end{claim}
Note that $\overline\alpha$ can be interpreted as the probability that the campaign goes viral when started from a randomly chosen pioneer. 

In the Appendix we sketch the main arguments allowing to prove the above claims; see~\cite{viral} for  formal statements and proofs.
 Recall also from~\cite{JanLuc} that under assumption~(\ref{e.CA}) the size $|C_0|$ of the big
network component $C_0$ 
satisfies for $n$ large
\begin{equation}
 \frac{\left|C_0\right|}{n} \approx 1-G_D(\xi_0) =: \alpha_0 > 0\,,
\end{equation}
where $\xi_0$ is the unique zero of 
$$H_0(x):=\E[D]x^2-xG'_D(x)$$ 
in $(0,1)$, with $G'_D(x)$ denoting the derivative of the pgf  of~$D$.

\section{Examples}

Let us consider the results of Section~\ref{s.Model} 
in the context of a  few illustrative network examples.
\subsection{Bernoulli transmissions}
\label{ss.Bern}
Let us assume some arbitrary distribution of the degree $D$
satisfying~(\ref{e.CA}) (to guarantee the existence of the big
component of the social graph). Suppose that each user decides
independently for each of its friends with probability $p\in[0,1]$
whether to transmit the influence to him or not. 
We call this model {\em CM with Bernoulli
transmissions} and $p$ the  {\em transmission probability}.   Note that given the total degree $D$, the 
transmitter degree $D^t$ is Binomial$(D,p)$ random variable. 

\begin{fact}
\label{f.CMB}
In the CM with a general degree distribution $D$ satisfying~(\ref{e.CA}) and  Bernoulli transmissions,   the campaign can go viral if and only if the transmission probability $p$ satisfies 
\begin{equation}
\label{e.CMB-cond}
p>\frac{\E[D]}{\E[D^2]-\E[D]}\,.
\end{equation}
In this latter case the fraction of the influenced population  and the fraction of good pioneers  are asymptotically  equal  to each other  
$|C^*|/n\approx |\bar C^*|/n=:\alpha$, for large $n$, and satisfy 
\begin{equation}
\label{e.CMB-size}
\alpha=1-G_D(\xi)\,,
\end{equation}
where $\xi$ is the unique zero of the function 
$$\E[D]((x-1)/p+1)-G'_D(x)$$ in $(0,1)$.
\end{fact}
\begin{proof}
Bernoulli transmissions  with~(\ref{Hx}) and~(\ref{Hbarx}) imply $H(x)=\mathbb{E}[D] x^2-(1-p)\mathbb{E}[D] x-pxG'_D(x)$
and  $\overline{H}(x)=\mathbb{E}[D] x^2-G'_D(1-p(1-x))$.
Moreover $\overline G_{D^{(t)}}(x)=G_{D}(1-p(1-x))$. 
Dividing $H(x)$ by $px$ and substituting $y:=1-p(1-x)$ in $\overline H(x)$ and $\overline G_{D^{(t)}}(x)$ completes the proof.
\end{proof}

Consider two specific network degree examples.
\begin{exe}[Poisson degree]
When $D$ has Poisson distribution of parameter $\lambda$
(in which case the CM is asymptotically equivalent to the
Erd\"os-R\'enyi model) the condition~(\ref{e.CMB-cond}) reduces to 
$$\lambda p>1$$
 and the fraction of the influenced population and good pioneers~(\ref{e.CMB-size}) 
is equal to 
$$\alpha=(1-\xi)/p\,,$$
 where 
 $\xi$ is the unique zero of the function 
$$(x-1)/p+1-\exp(\lambda(x-1)$$ in $(0,1)$\,.
\end{exe}

More commonly observed degree-distributions
in social networks have power-law tails.

\begin{exe}[Power-Law (``zipf'') degree]
Assume $D$ having distribution 
$$\P\{D=k\}=k^{-\beta}/\zeta(\beta)\quad k=1,2,\ldots\,,$$
with $\beta>2$, where $\zeta(\beta)$ is the zeta function. 
Recall that the pgf of $D$ is equal to $G_D(x)=\text{Li}_\beta(x)/\zeta(\beta)$, where $\text{Li}_\beta(x)=\sum_{k=1}^\infty k^{-\beta}x^k$ is the so-called poly-logarithmic function.
Condition~(\ref{e.CA}) for the existence of the big component
is equivalent to 
$$\zeta(\beta-2)-2\zeta(\beta-1)>0\,,$$
which is  approximately $\beta<3.48$.
Condition~(\ref{e.CMB-cond}) reduces to 
$$p>\zeta(\beta-1)/(\zeta(\beta-2)-\zeta(\beta-1))$$
and the fraction of the influenced population and good pioneers~(\ref{e.CMB-size})  
is equal to  
$$\alpha=1-\text{Li}_\beta(\xi)\,,$$ where 
$\xi$ is the unique  zero of the function 
$$x\zeta(\beta-1)((x-1)/p+1)-\text{Li}_{\beta-1}(x)$$ in $(0,1)$.
\end{exe}

Recall from Fact~\ref{f.CMB}, that Bernoulli transmissions lead to the model  where 
the fraction of the influenced population  and the fraction of good pioneers  are asymptotically  equal  to each other.  In what follows we present two scenarios  where the set of good pioneers
and the influenced population have different size.

\subsection{Enthusiastic and apathetic users or node percolation} 
\label{ss.percol}
Consider CM with a general degree distribution $D$ satisfying~(\ref{e.VA}), whose nodes
either transmit the influence to all their friends (these are ``enthusiastic'' nodes) or do not transmit to any of their friends (``apathetic'' ones). Let $p$ denote the fraction of nodes in the network which are enthusiastic. Note that this model corresponds to the {\em node-percolation}~\footnote{different than edge-percolation} on the CM. Thus, in this model, given $D$,
 $D^{(t)}=D$ with probability $p$ and $D^{(t)}=0$ with 
probability $1-p$.

\begin{fact}
Consider node-percolation on the CM with a general degree distribution $D$ satisfying~(\ref{e.CA}). The campaign can go viral if and only if the fraction $p$ of enthusiastic users satisfies condition~(\ref{f.CMB}); the same as for the Bernoulli model. Moreover, in this case,
the fraction $\alpha$ of reached population is also the same as in the network with  Bernoulli transmissions, i.e., 
equal to~(\ref{e.CMB-size}) with $\xi$ as in Fact~\ref{f.CMB}. However, the fraction $\overline\alpha$ of good pioneers is equal to $\overline\alpha=p\alpha$.
\end{fact}
The proof  follows easily from the general results of Section~\ref{ss.Claims}. Note that the {\em campaign on the network with enthusiastic and apathetic users can reach the same population as in the Bernoulli transmissions, however there are less good pioneers}.

\subsection{Absentminded users or coupon-collector transmissions}  
\label{ss.CC}
Consider again CM with a general degree distribution $D$ satisfying~(\ref{e.VA}). Suppose that each user is willing (or allowed) to transmit $K$ messages of influence. In this regard, it randomly selects $K$ times one of his friends {\em with replacement}
(as if he were forgetting his previous choices).
An equivalent dynamic of the influence propagation 
can be formulated as follows: every influenced user, at all times, keeps choosing  one of its friends
uniformly at random and transmits the influence to him;
it stops forwarding the influence  after $K$ transmissions.

In this model the transmission degree $D^{(t)}$ 
correspond to  the number of collected coupons in the classical coupon collector problem with the number of coupons being the  vertex degree $D$ and the
number of trials $K$.
The conditional distribution of $D^{(t)}$ given $D$
can be expressed as follows:
$$
\P\{\,D^{(t)}=k\,|\,D\,\}=\frac{D!}{(D-k)!D^{-K}}\Bigl\{{K\atop k}\Bigr\}\,,$$
where $\{{K\atop k}\}=1/k!
\sum_{i=0}^K(-1)^i{k\choose i}(k-i)^K$ is the Stirling number of the second kind. 

Calculating the pgf for this distribution is tedious and we 
do not present analytical results regarding this model but only
simulations and estimation. As we shall see in
Section~\ref{ss.Numerical}, in this model 
 {\em the influenced population is smaller than the population of good pioneers.}

\subsection{Numerical examples}
\label{ss.Numerical}
We will present now a few numerical examples of networks and 
diffusion models presented above.

\subsubsection{Simulations}
In all our examples we simulate the enhanced configuration model on
$N=1000$ nodes
 assuming some particular 
node degree $D$ distribution and influence propagation mechanism
modeled by the conditional distribution of the transmitter degree
$D^{(t)}$. More precisely, we sample the individual node degrees and
transmitter degrees  $(D_i,D_i^{(t)})$ 
$i=1\,...\,N$ independently from  the joint distribution of 
 $(D,D^{(t)})$ 
and use these values to construct an instance of our 
enhanced CM by uniform pairwise matching of the half-edges.
We calculate the relative size of the  influenced population and the
set of good pioneers through the exploration of the
influenced components for all nodes.~\footnote{The simulations are run in
{\em python} using the {\em networkx} package.}
In fact, relative sizes of the
populations reached from different pioneers concentrate 
very clearly, as shown on Figure~\ref{f.Hist}, which illustrates the
statement of Claim~\ref{cl1}.  

\begin{figure}[h!]
\begin{center}
\centerline{\includegraphics[width=1\linewidth]{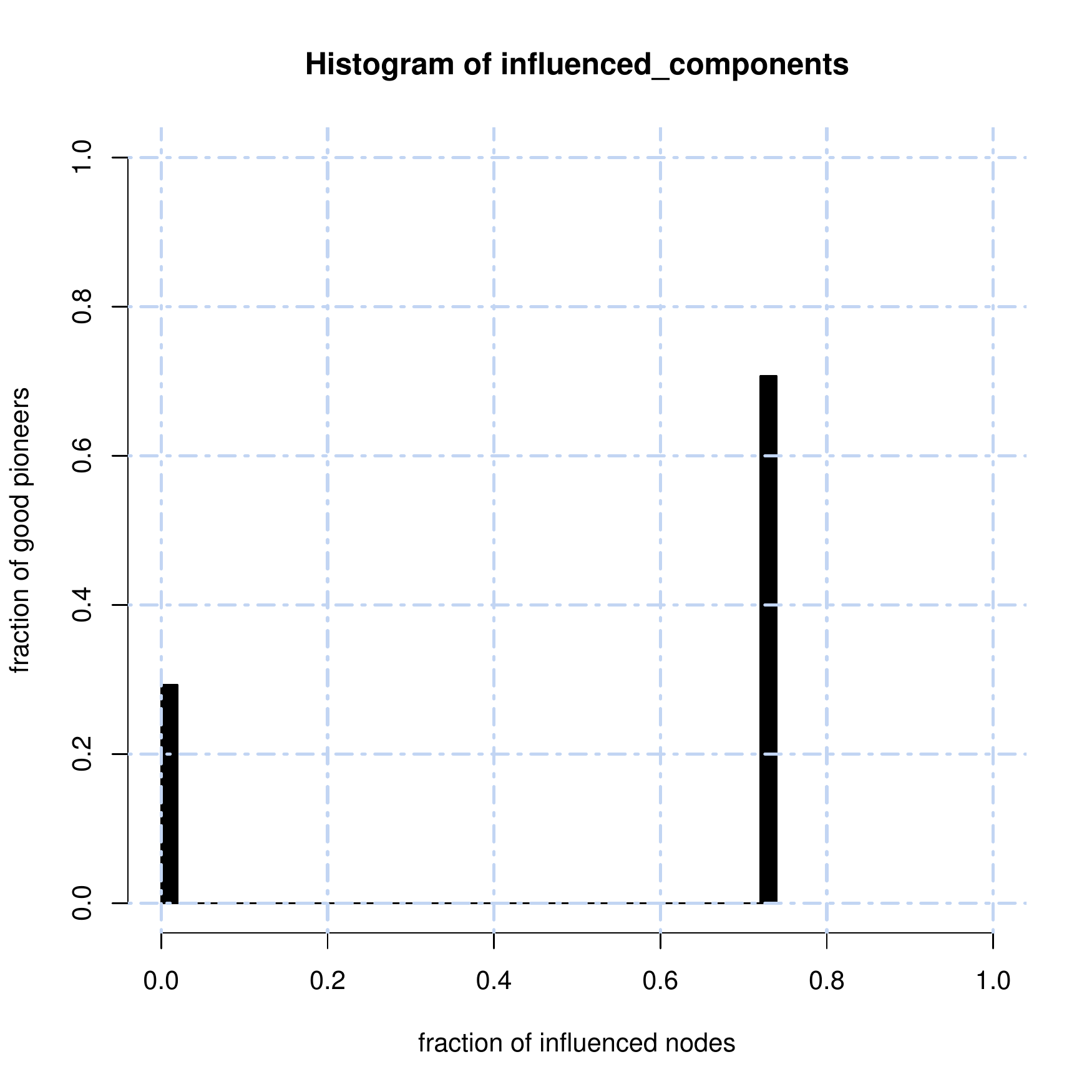}}
\caption{Concentration of the relative sizes of 
populations reached from different pioneers. CM with Poisson degree of mean
$\lambda=2$ and Bernoulli transmissions with $p=0.8$.
 \label{f.Hist}}
\end{center}
\end{figure}

\subsubsection{Estimation}
\label{sss.Estimation}
We adopt also the following ``semi-analytic'' approach:
Using the sample $(D_i,D_i^{(t)})$, $i=1,\ldots,N$ used to
construct the CM, we consider estimators 
\begin{align}
\hat G_D(x)&:=\frac{1}{N}\sum_{i=1}^N x^{D_i}\label{e.hatD}\\
\hat G_D^{(t)}(x)&:=\frac{1}{N}\sum_{i=1}^N x^{D_i^{(t)}}\label{e.hatDt}\\
\hat H(x)&:=\frac{1}{N}\sum_{i=1}^N \Bigl( D_i x^2-(D_i-D_i^{(t)})x-
D_i^{(t)}x^{D_i}\Bigr)\label{e.hatH}\\
\hat{\overline{H}}(x)&:=\sum_{i=1}^N \Bigl(D_i x^2-D_i^{(t)}x^{D_i^{(t)}} - 
(D_i-D_i^{(t)})x^{D_i^{(t)}+1}\Bigr)\,\label{e.hatHbar}
\end{align}
of the functions $G_D(x)$,  $G_{D^{(t)}}(x)$, $H(x)$ and
 $\overline H(x)$, respectively. We calculate estimators $\hat\alpha$
 and $\hat{\overline\alpha}$ of the fraction of the influenced population
 $\alpha$ and of  good pioneers $\overline\alpha$
using Claim~\ref{e.alpha} and the estimated functions $\hat G_D(x)$,
$\hat G_{D^{(t)}}(x)$, $\hat H(x)$ and
 $\hat{\overline H}(x)$. (That is, we  find numerically zeros $\hat\xi$ and
 $\hat{\overline\xi}$ of  $\hat H(x)$  and $\hat{\overline H}(x)$,
 respectively, and plug them into~(\ref{e.alpha}) and~(\ref{e.alphabar}),
with $\hat G_D(x)$ and $\hat G_{D^{(t)}}(x)$ replacing   $G_D(x)$ and
$G_{D^{(t)}}(x)$.)

Note that in the semi-analytic approach we do not need to know/construct 
the realization of the underlying model. This observation is
a basis  of a  {\em campaign evaluation method} that we propose in
Section~\ref{s.Implications}. In fact, in reality one usually 
does not have the complete insight into the network structure
and needs to rely on  statistics collected  from the  initially
contacted pioneers. 

\subsubsection{Analytic evaluation}
Finally, for all models, except the ``coupon-collector'' one of
Section~\ref{ss.CC}, we calculate numerically the values of $\alpha$
and $\overline\alpha$ using the explicit forms of all the involved
functions.
(For the coupon-collector model we obtained the ``true'' values of   $\alpha$
and $\overline\alpha$ from a sample of $(D_i,D_i^{(t)})$ of a
larger size $N$.)

When comparing these analytic solutions to the simulation and
semi-analytic estimates we see that in some cases $N=1000$ is not big
enough to match the theoretical values. One can easily consider 
larger samples, however we decided to stay with $N=1000$ to show how
the quality of the estimation varies over different model assumptions. 
Also, $N=1000$ seems to be near the lower range of the  number of initial
pioneers one needs to contact to produce a reasonable prognosis for the
development of the campaign.

\subsubsection{Case study}
Figures~\ref{f.PB} and~\ref{f.PLB}
present  Bernoulli influence propagation on the CM with Poisson and Power-Law degree distribution of mean $\E[D]=2,\,4,\,6$.
Bernoulli transmissions imply the set of good pioneers and influenced population of the same size.
The Power-Law degree with $\beta<3$ leads to positive fraction of good pioneers and influenced component for all $p>0$,
while for the Poisson degree distribution one observes the phase transition at $p=1/\lambda$.
That is,  the fractions of good pioneers and the influenced component are strictly positive if and only if $p>1/\lambda$.

Figure~\ref{f.PBPL} shows again the model with Bernoulli transmissions
on CM with Poisson and Power-Law degree distribution, this time
however for $\E[D]\approx 1.35$ for which both models exhibit the
phase transition in $p$.

A general observation is that the Power-Law degree distribution gives smaller
critical values of $p$ for the existence of a positive fraction of influenced
population  and good pioneers, however for these the size of these sets
increase with the transmission probability $p$ more slowly in the Bernoulli model.  
Obviously the values of $\alpha=\overline\alpha$ at $p=1$ correspond to the size of the biggest
connected component of the underlying CM.

Figure~\ref{f.PNPPL} shows  the node percolation (or ``apathetic and
enthusiastic users) on CM 
with Poisson and Power-Law degree distribution of mean 
$\E[D]\approx 2$. Note that the influenced components have the same size as
for Bernoulli transmissions, however good
components are smaller. The critical values of $p$ for the phase
transition are also the same as for Bernoulli transitions.
Note that estimation of the  node percolation model is more
difficult  than the Bernoulli transmissions because of higher variance of the estimators.

Finally, Figure~\ref{f.PCCPL} shows  that 
the coupon collector dynamics (``absentminded users'') on CM
produces  bigger sets of  good pioneers than  the influenced population.

\section{Application to Viral Campaign Evaluation}
\label{s.Implications}
What does the analysis presented up to now suggest in terms of strategy for a firm which is just about to start
a new marketing campaign on an online social network without having any prior information about
the network structure? 

If the fraction of good pioneers in the network is non-negligible, the 
firm has a strictly positive probability of picking a good pioneer even when it picks a pioneer 
uniformly at random from the network. 
Now when is the fraction of good pioneers non-negligible?
Since the firm has no prior information about the network structure and the campaign effectiveness,
the best it can do is to collect information from its pioneers regarding the number of friends that they
have (total degree) and the number of friends they influence in this campaign (transmitter degree), and 
then assume that the network is a uniform random network having the sampled total degree and transmitter
degree distributions. The collected information, denote it by  $(D_i,D_i^{(t)})$, $i=1,\ldots, N$, allows to 
estimate various quantities relevant to the potential development of the ongoing campaign,  as we did in~\ref{ss.Numerical}.

More precisely the results presented in Section~\ref{ss.Claims} suggest the following 
approach.

\paragraph{Network fragmentation}
The first and foremost question  is whether 
the  network is not too fragmented to allow for viral marketing. 
This is related to condition~(\ref{e.CA}). In order to answer this question 
one considers the following estimator of $\E[D^2-2D]$
$$\frac{1}{N}\sum_{i=1}^N \Bigl(D_i^2-2D_i\Bigr)\,.$$
 If the value of this estimator is not sharply larger than zero then
the firm must assume that the network is too fragmented to allow for viral marketing. 
 Natural confidence intervals can be considered in this context too. Evidently, the confidence increase as the firm picks more pioneers and collects more data.

\paragraph{Effectiveness of the campaign}
If one estimates that the network is not too fragmented, 
then the firm can evaluate the effectiveness of the ongoing campaign.
It is related to condition~(\ref{e.VA}).
Again one considers the natural estimate of 
 $\E[DD^{(t)}-D^{(t)}-D]$
$$\frac{1}{N}\sum_{i=1}^N \Bigl(D_iD_i^{(t)}-D_i-D_i^{(t)}\Bigr)\,.$$
If the value of this estimator is sharply larger than zero then
the firm can assume that there is  a realistic chance of picking
a good pioneer via random sampling and make the campaign go viral.
Otherwise,  the previous phase of the campaign  can be considered as non-efficient.

\paragraph{Cost-benefit analysis} 
 If the firm deems the campaign
to be effective, it can then, exactly as we did in \ref{sss.Estimation}, come up with the estimates of the relative fractions of good pioneers and population vulnerable to influnce, and do a cost-benefit analysis.
What we have described is an outline which can be used by the firms to come up with a rational methodology
for making decisions in the context of viral marketing.

\begin{figure}[t!]
\begin{center}
\centerline{\includegraphics[width=1.02\linewidth]{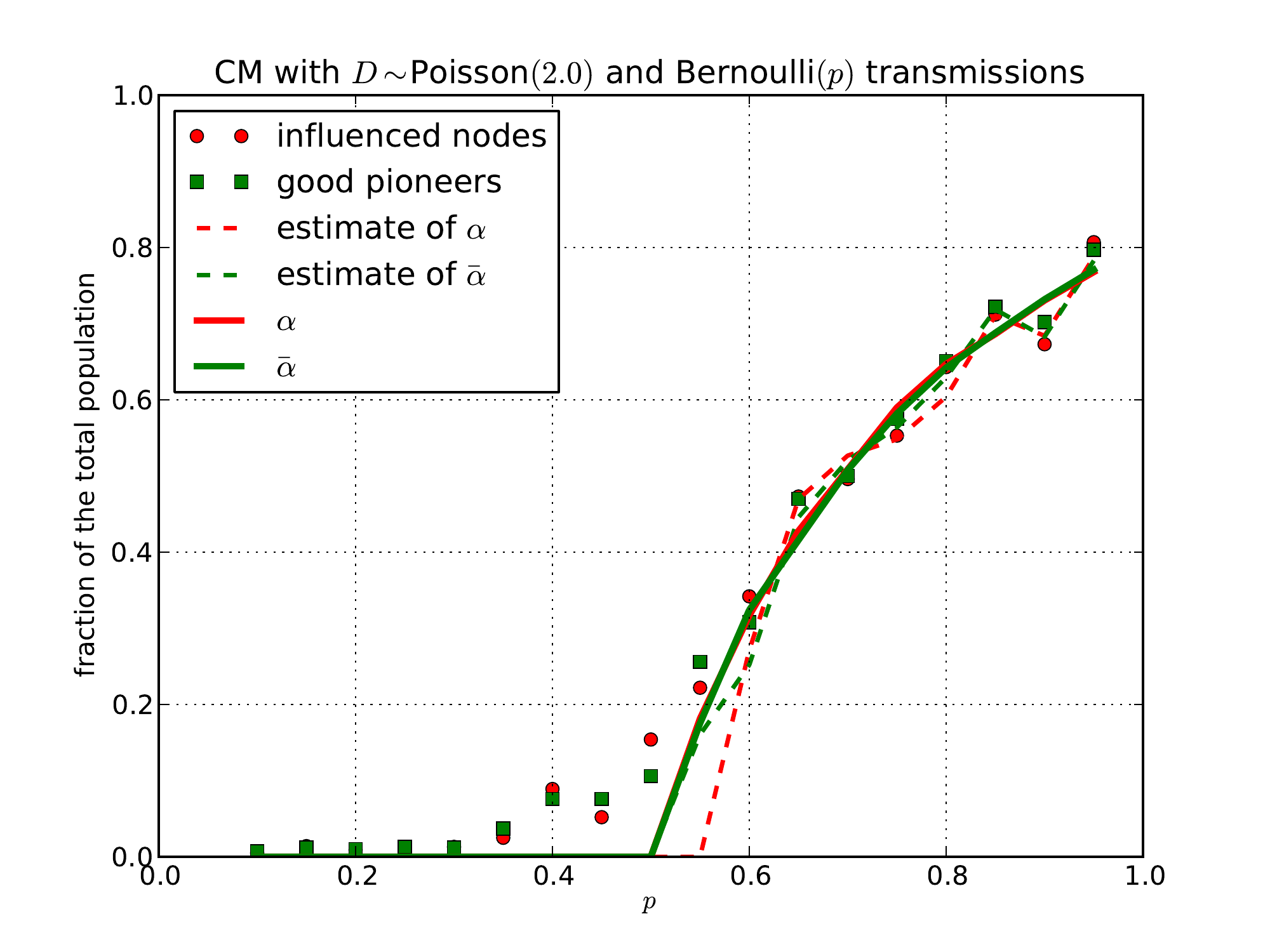}}
\centerline{\includegraphics[width=1.02\linewidth]{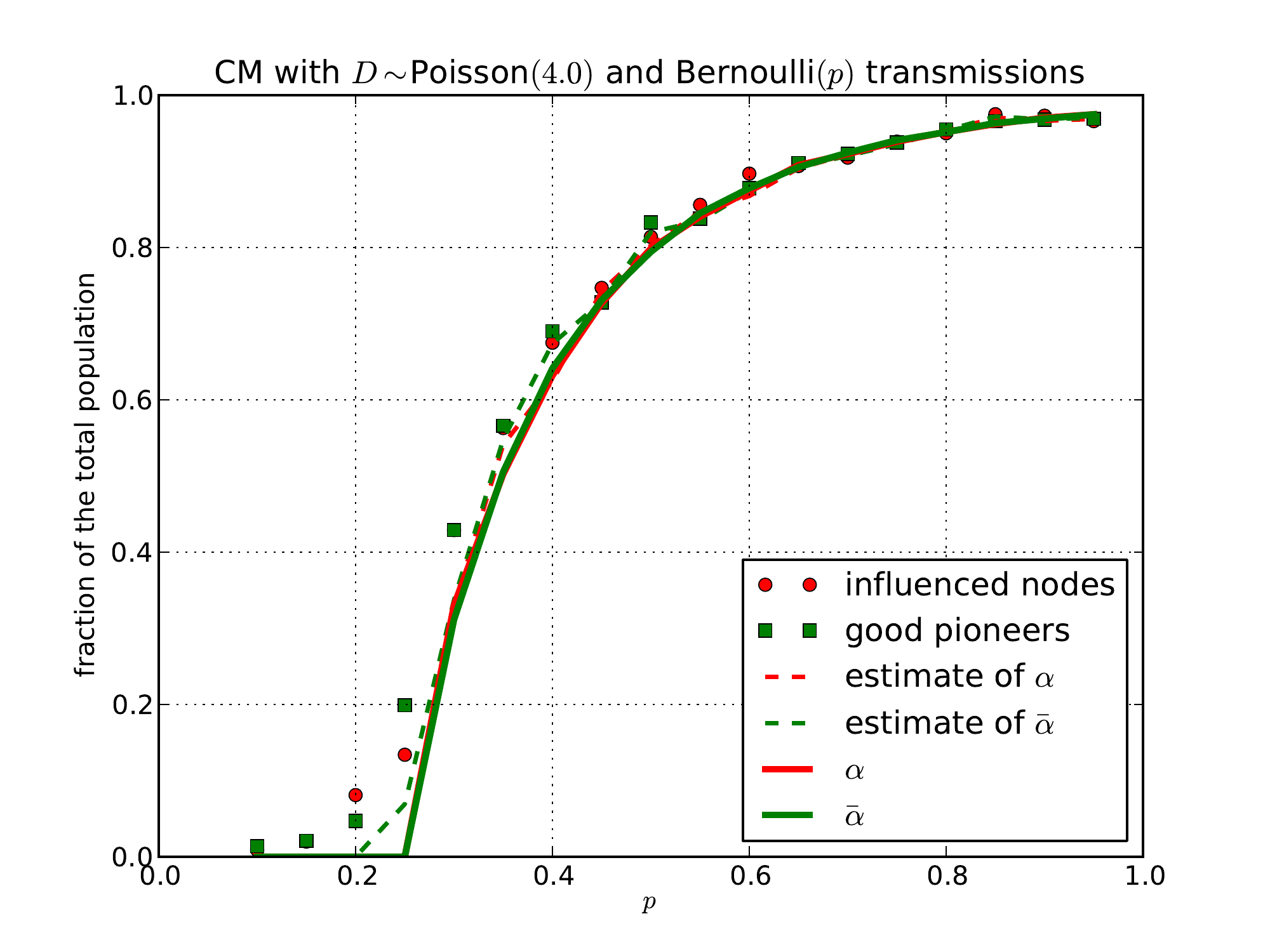}}
\centerline{\includegraphics[width=1.02\linewidth]{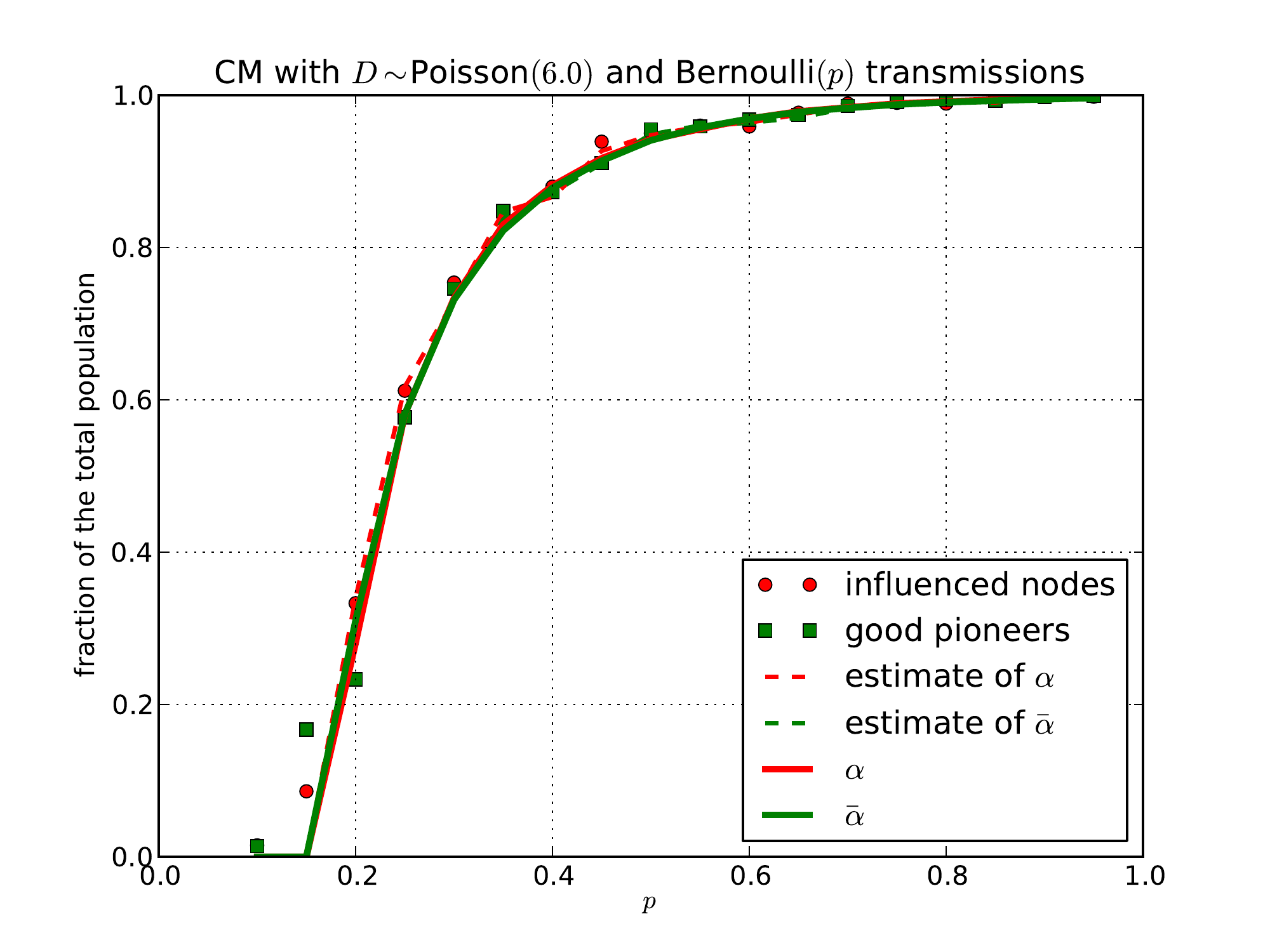}}
\caption{CM with Poisson degree of mean $\lambda=2,\,4,\,6$ and
  Bernoulli  transmissions with probability $p$.
The set of good pioneers and the influenced population are of the same size. Their fraction is strictly positive for $p>1/\lambda$.
 \label{f.PB}}
\end{center}
\end{figure}

\begin{figure}[t!]
\begin{center}
\centerline{\includegraphics[width=1.02\linewidth]{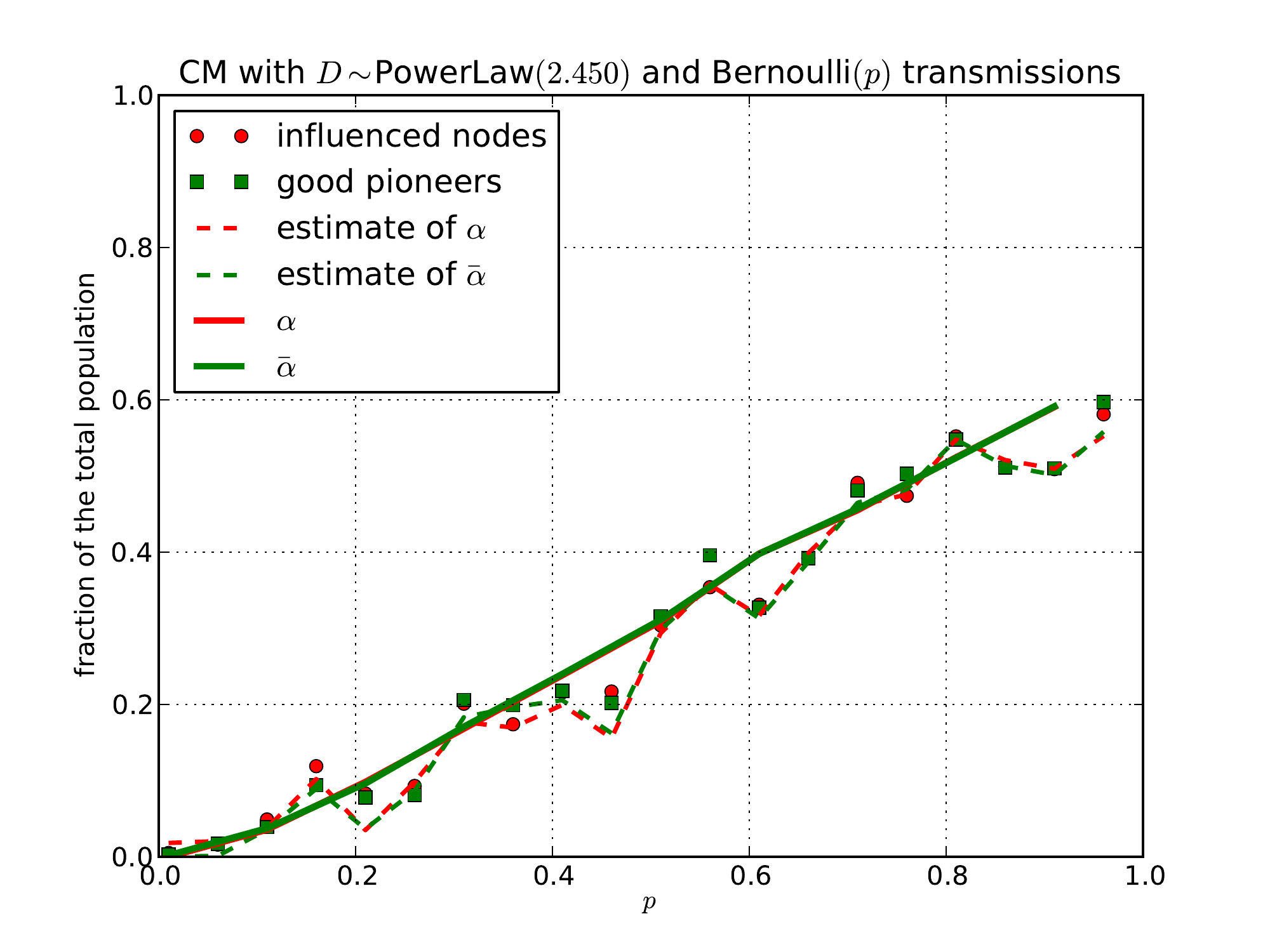}}
\centerline{\includegraphics[width=1.02\linewidth]{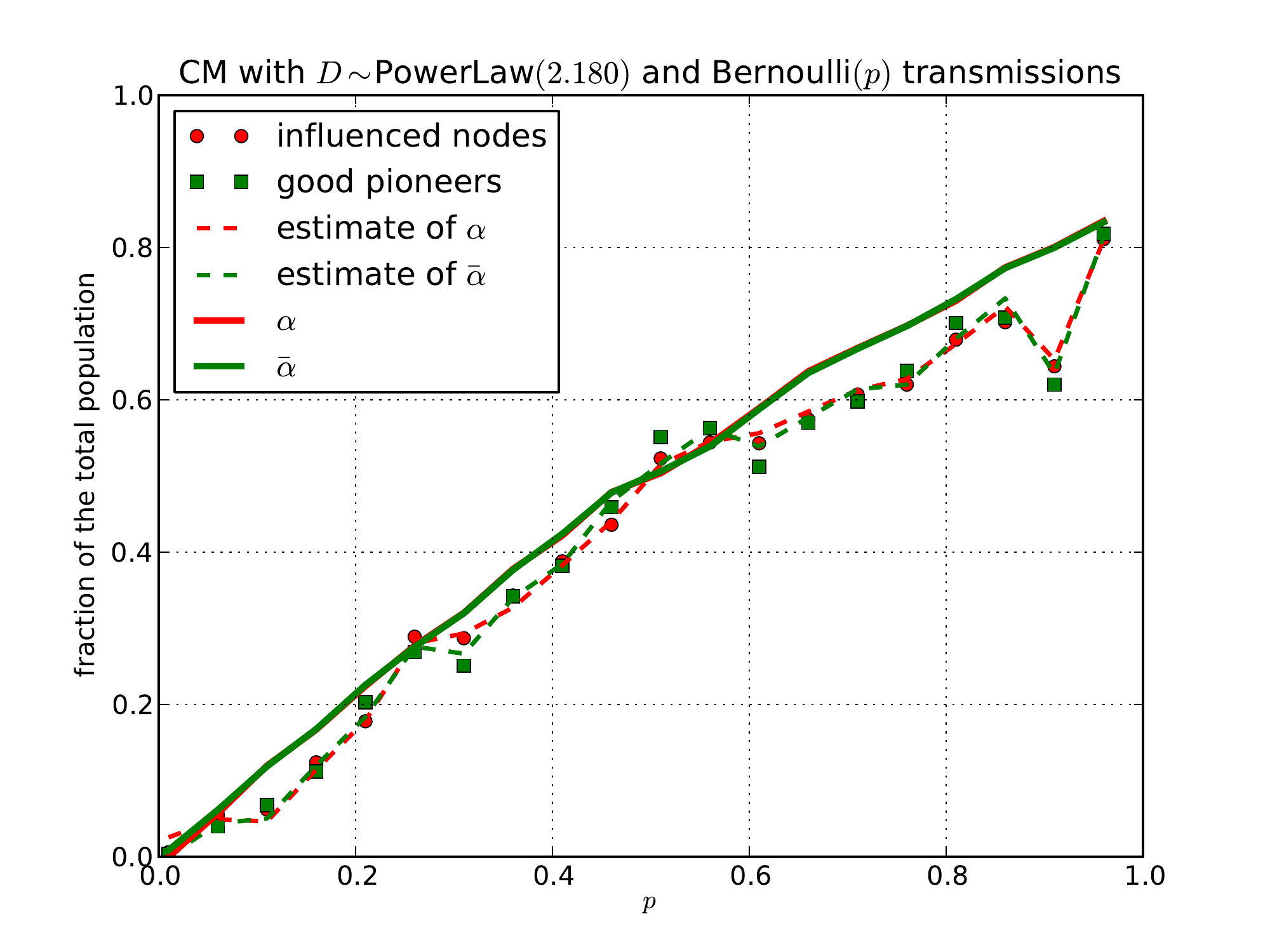}}
\centerline{\includegraphics[width=1.02\linewidth]{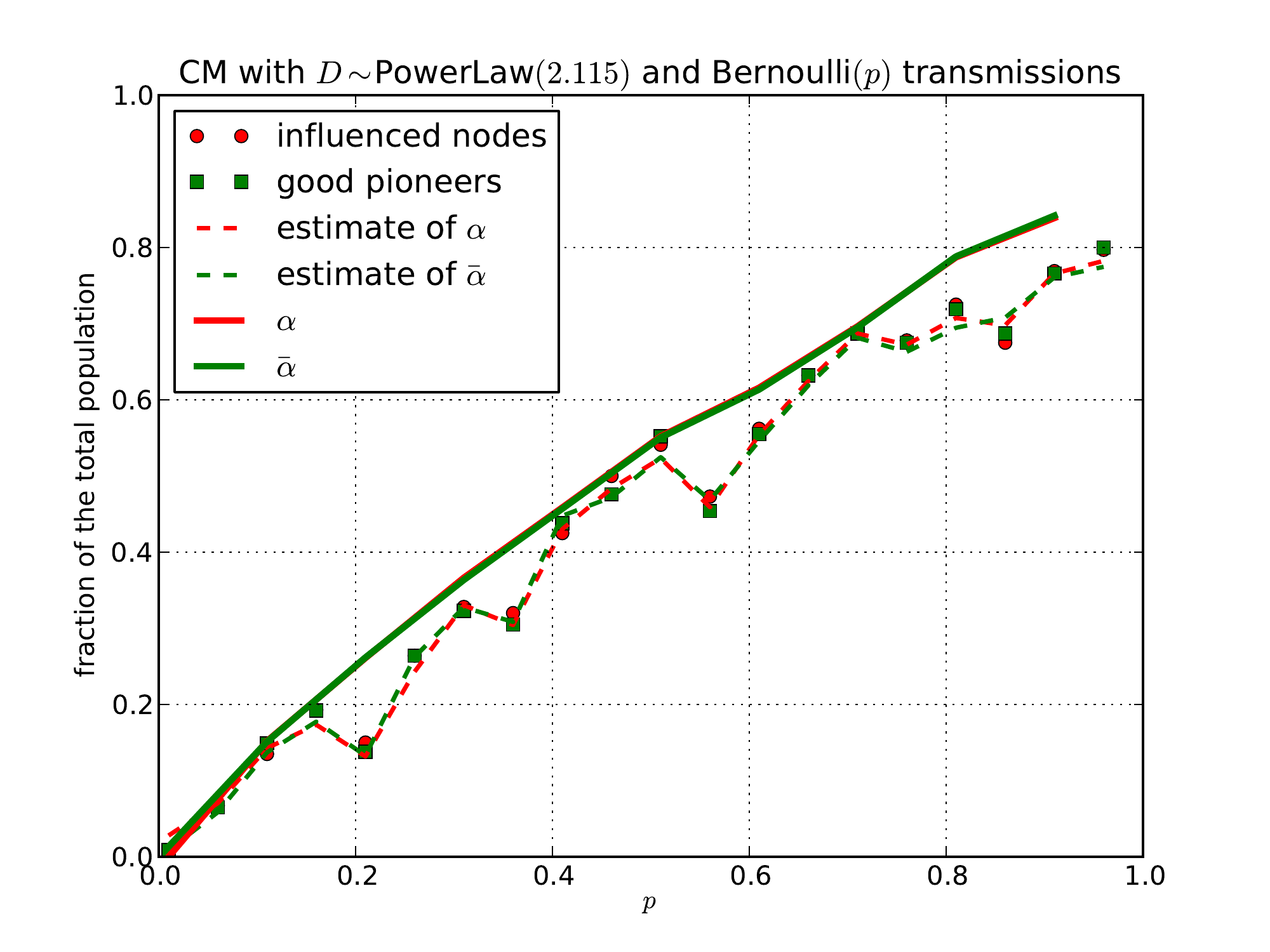}}
\caption{CM with Power-Law degree of parameter $\beta=2.450,\,2.180,\,2.115$ (corresponding  to $\E[D]\approx 2,\,4,\,6$ and
  Bernoulli  transmissions with probability $p$. 
The set of good pioneers and the influenced population are of the same size. Their fraction  is strictly positive for all $p>0$
whenever  $\beta\le 3$.
 \label{f.PLB}}
\end{center}
\end{figure}

\begin{figure}[t!]
\begin{center}
\centerline{\includegraphics[width=1.2\linewidth]{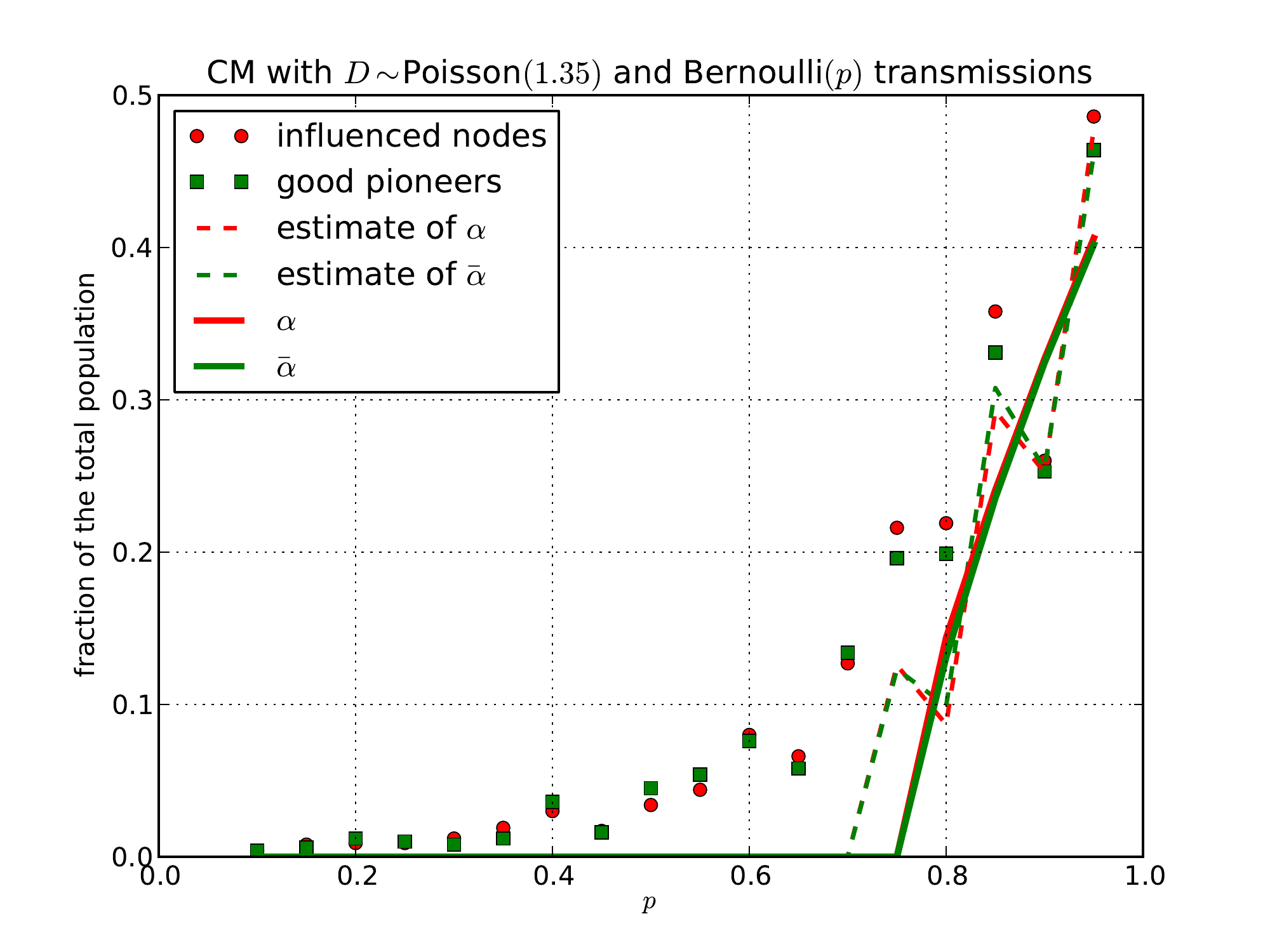}}
\centerline{\includegraphics[width=1.2\linewidth]{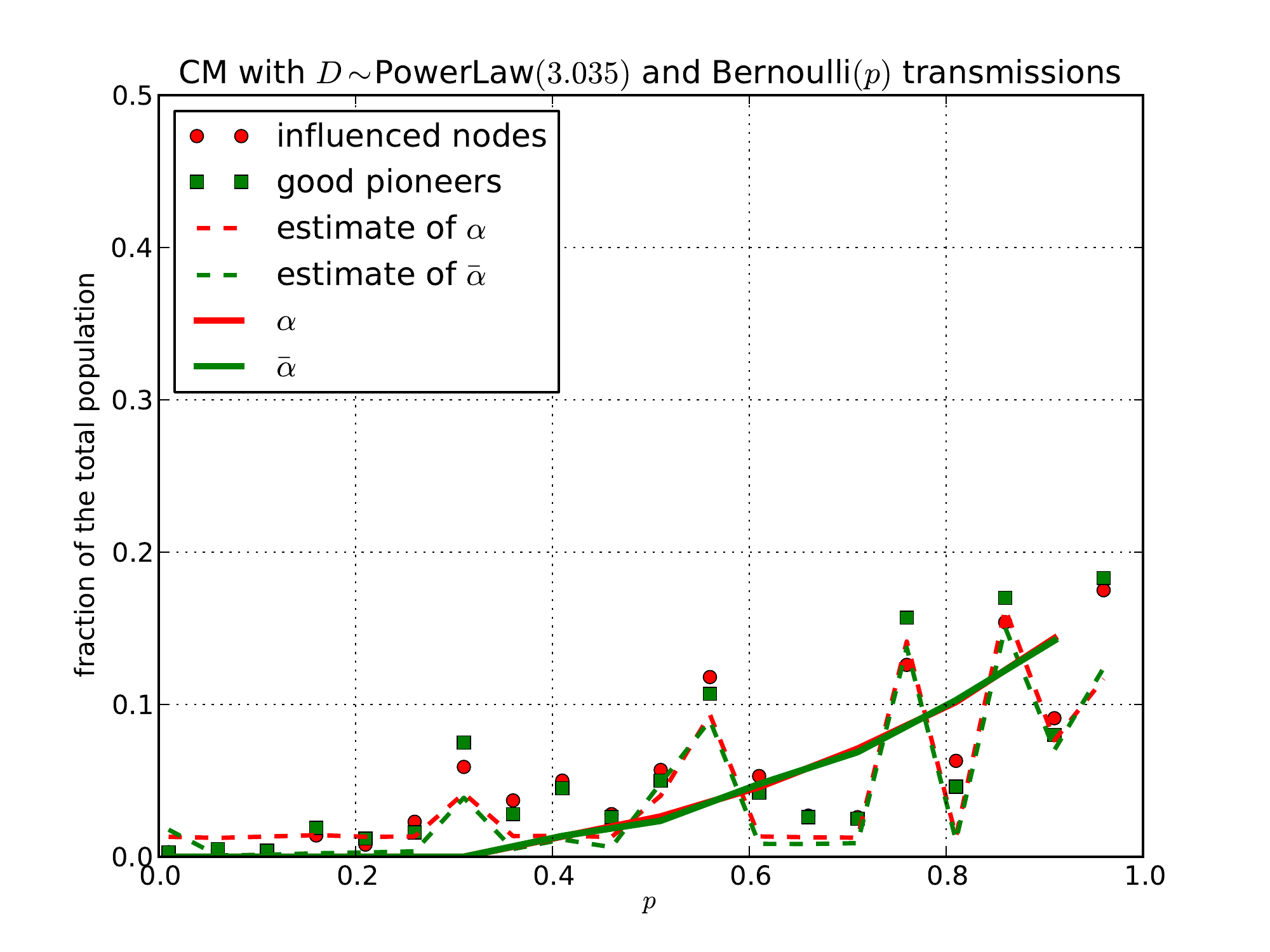}}
\caption{CM with Poisson and Power-Law degree  of mean $\E[D]\approx
  1.35$ ($\lambda=1.35$ and $\beta=3.035$) and Bernoulli
  transmissions. The set of  good pioneers and the influenced
  population are of the same size for each model. One observes the phase transition in both models, at $p=1/\lambda$ and $p=\zeta(\beta-1)/(\zeta(\beta-2)-\zeta(\beta-1)$, respectively. 
 \label{f.PBPL}}
\end{center}
\end{figure}

\begin{figure}[t!]
\begin{center}
\centerline{\includegraphics[width=1.2\linewidth]{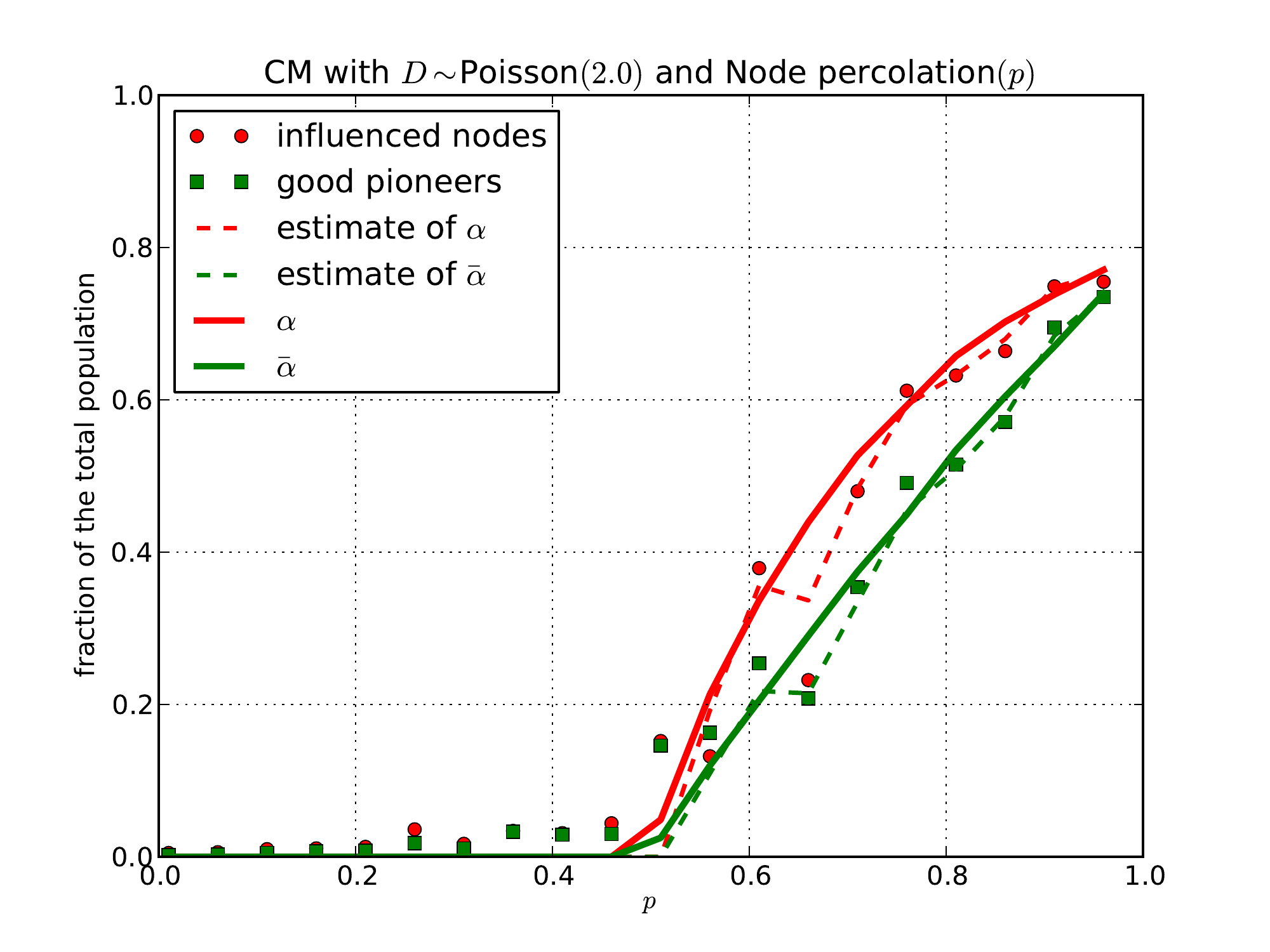}}
\centerline{\includegraphics[width=1.2\linewidth]{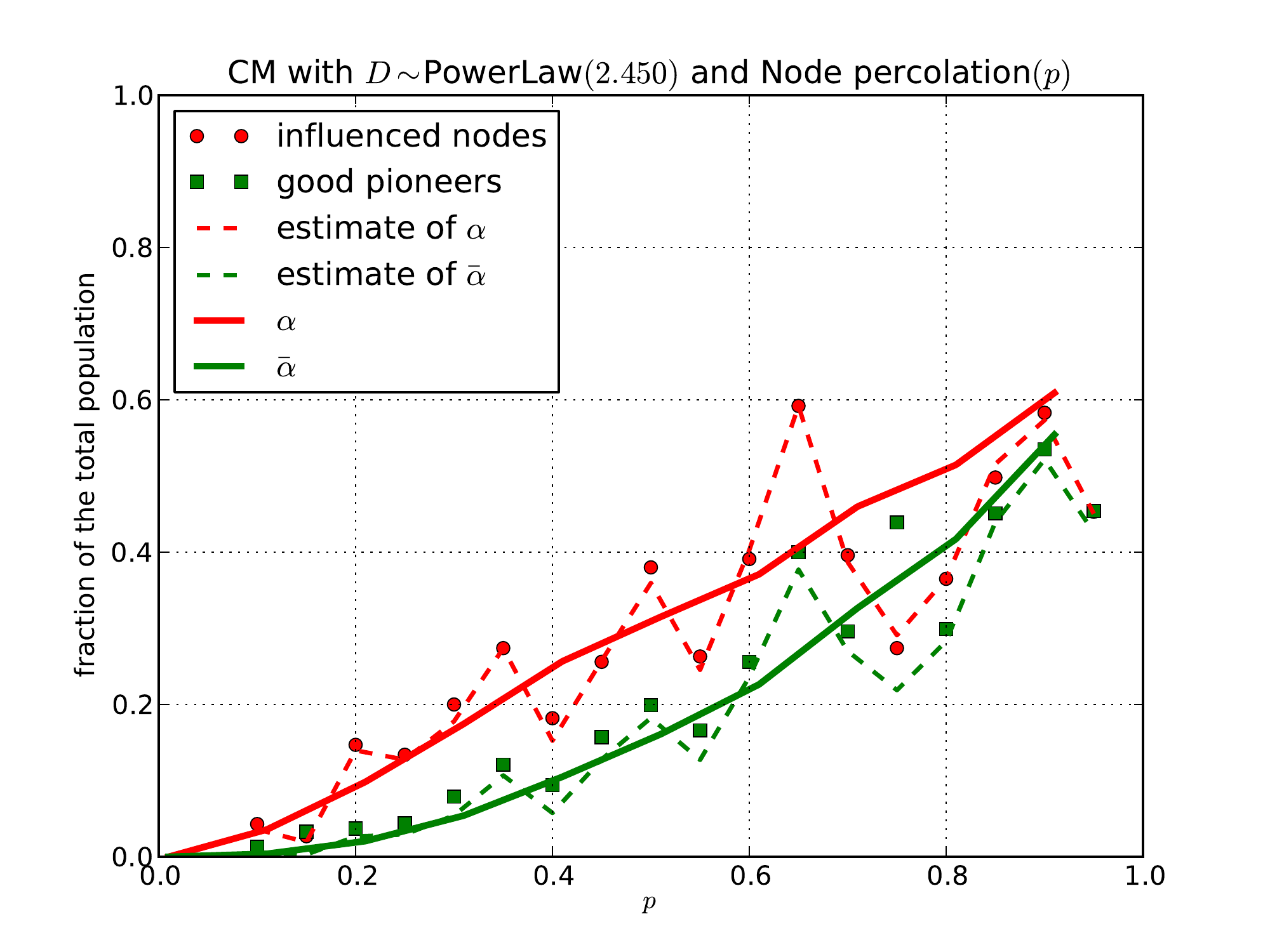}}
\caption{Node percolation (``apathetic and enthusiastic users'') on CM
  with Poisson and Power-Law degree  of mean $\E[D]\approx 2$
  ($\lambda=2$ and $\beta=2.45$).  The influenced component and the critical values for $p$ 
 are equal to these for the CM with Bernoulli transmissions.
 The set of  good pioneers is  smaller than  the influenced population.
 We do not observe the  phase transition for the Power-Law model since $\beta<3$.
 \label{f.PNPPL}}
\end{center}
\end{figure}

\begin{figure}[t!]
\begin{center}
\centerline{\includegraphics[width=1.2\linewidth]{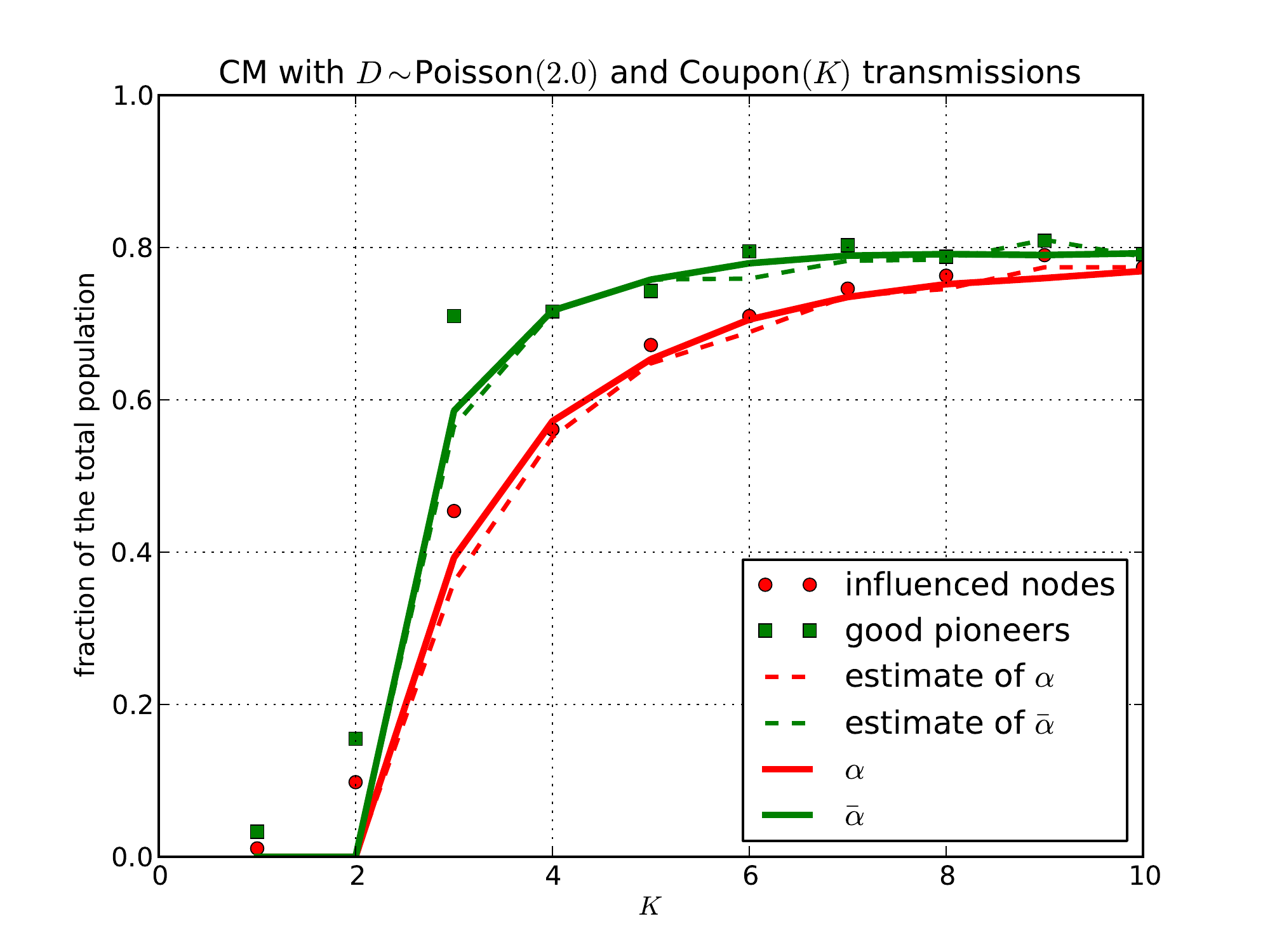}}
\centerline{\includegraphics[width=1.2\linewidth]{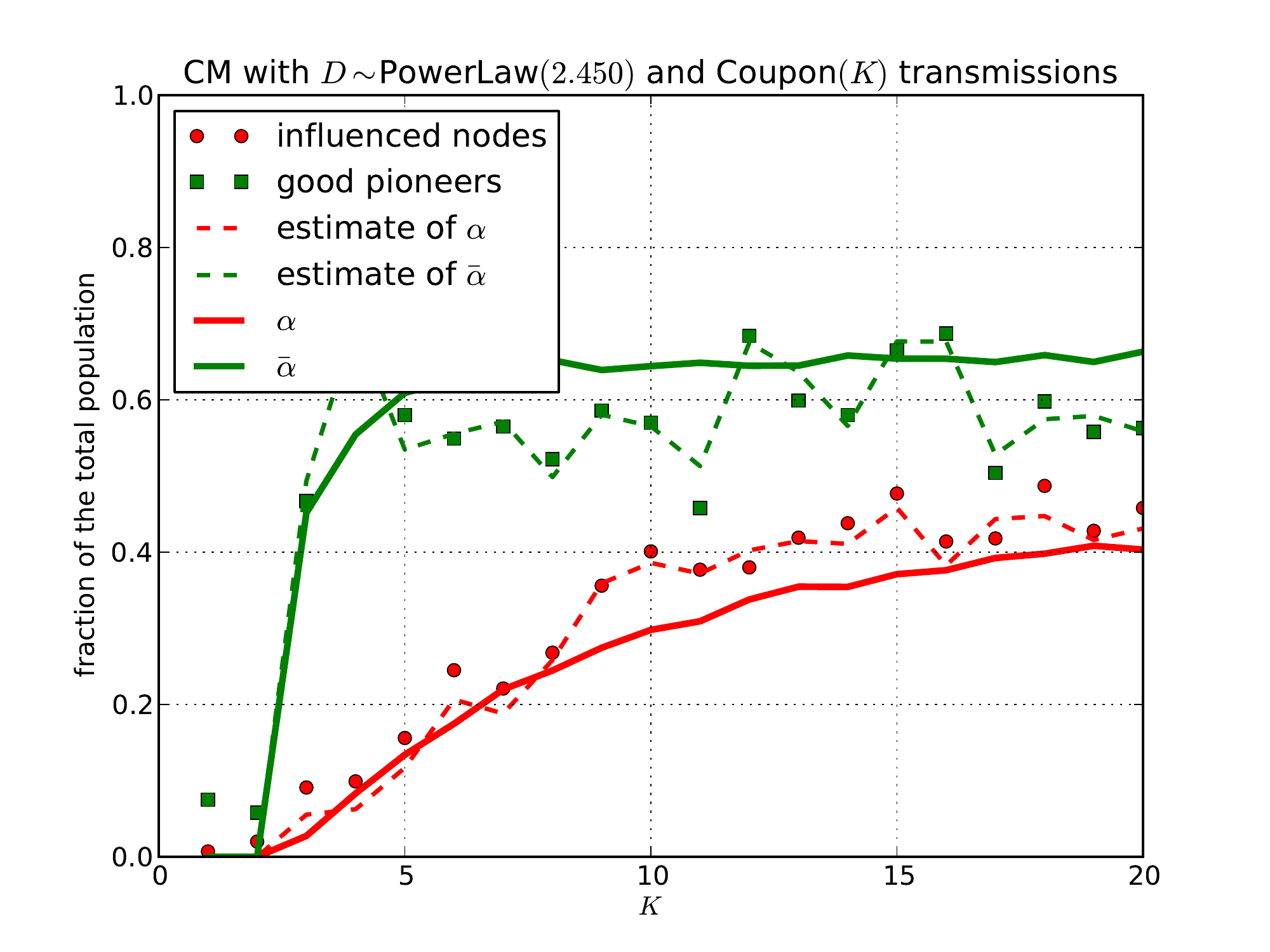}}
\caption{Coupon collector dynamics ``absentminded users'') on CM
  with Poisson and Power-Law degree  of mean $\E[D]\approx 2$
  ($\lambda=2$ and $\beta=2.45$). The set of  good pioneers is
  bigger than  the influenced population. 
 \label{f.PCCPL}}
\end{center}
\end{figure}

\section{Conclusion and Future Work}
Diffusion studies on networks generally tend to focus on the component
of population which is vulnerable to
influence, the component that can be reached starting from an initial
target.  In this work, we focus on the 
other side of coin, i.e., the subset of population, called good
pioneers, from which an
initial target must be picked so that a large fraction of the
population is influenced.

Our  analysis of the set of good pioneers 
is based on a new approach proposed in~\cite{viral},
consisting of identifying this subset as the big component of a  {\em
  reverse dynamic} 
in which  an ``acknowledgement'' message 
is sent in the reversed direction on every edge thus allowing to  trace 
all the possible sources of influence of
a given vertex. 

Based on the recent graph-theoretical results
obtained through this approach regarding the existence and the size of both
subsets: of good pioneers and vulnerable population, we propose a simple yet useful methodology for the analysis
of an ongoing viral marketing campaign
from the statistical data gathered at its early stage.
It allows to verify whether the network is not too fragmented,
check the effectiveness of the previous phase of the campaign,
and  gives tools for making rational economic decisions regarding its future
development.

Among many  interesting questions raised by the present work, 
let us mention the relationship between connectedness of the sub-graph
induced by good pioneers and the measures of centrality for viral marketing.

\appendix

In what follows we briefly present the main arguments
leading to the results presented in Section~\ref{ss.Claims}.
See~\cite{viral} for the details.
\section{Influence diffusion analysis}
A standard technique 
for the analysis of diffusion of information on the CM
involves the simultaneous exploration of the model and the propagation
of the influence.
This exploration-and-propagation process can be approximated at two
different ``scales'':
A {\em branching process} can be used to 
approximate the initial phase of the process.
In fact, the  CM is locally ``tree-like'' and probability of non-extinction
of this tree should be related to the probability of choosing a good
pioneer for the influence propagation.
A {\em fluid limit} analysis can be used to describe the evolution
of the  process up to the time when 
the exploration of a  big component is completed and to characterize
the size of this component. 
The latter approach, recently proposed in~\cite{JanLuc},
is adopted in~\cite{viral} to prove the results presented in
Section~\ref{ss.Claims}.

A fundamental difference with respect to the study of the
big component of the classical CM stems from the directional
character of our propagation dynamic:
the edges matching a transmitter and a receiver half-edge
can relay the influence from the transmitter half-edge to the receiver one,
but not the other way around. This means that  the good pioneers do not need
to belong to the big (influenced) component, and vice versa. 

In this context,  a {\em reverse dynamic} is introduced
in~\cite{viral}, in which 
a message (think of an ``acknowledgement'') can be sent in the
reversed direction on every edge (from an arbitrary half-edge to the
receiver one),  which traces all the possible sources of influence of
a given vertex.  
This reversed dynamic can be studied using the same fluid-limit approach as the
original one, leading to the proof of uniqueness
of the big component of the reversed process and the characterization of
its size. It is shown in~\cite{viral} that 
this reverse-dynamic approach precisely 
coincides with the
probability of the non-extinction of the branching process
approximating the initial phase of the original (forward) exploration
process.
That is (under some additional technical assumptions, which we expect
can be relaxed), the  big component of the reverse process
coincides with the set of good pioneers.

In what follows we briefly present the details of both dynamics as
well as the branching process approximation.

\subsection{Analysis of the forward-propagation process}
 Throughout the construction and propagation process, we keep track of what we call \textit{active transmitter} half-edges. To begin  with, 
 all  the vertices and the attached half-edges are \textit{sleeping} but once influenced, a vertex and its half-edges become \textit{active}. 
 Both 
 sleeping and active half-edges at any time constitute what we call \textit{living} half-edges and when two half-edges are matched to 
 reveal an 
 edge along which the flow of influence has occurred, the half-edges are pronounced \textit{dead}. Half-edges are further classified 
 according to 
 their ability or inability to transmit information as \textit{transmitters} and \textit{receivers} respectively. 
 We initially give all the half-edges i.i.d. random maximal lifetimes
 with exponential (mean one) distribution, then go through the following algorithm.
 
 \begin{enumerate}
  \item[C1] If there is no active half-edge (as in the beginning), select a sleeping vertex and declare it active, along with all its 
  half-edges. 
 For definiteness, we choose the vertex uniformly at random among all sleeping vertices. If there is no sleeping 
 vertex left, the process stops. 
  \item[C2] Pick an active transmitter half-edge and kill it. \label{c2}
  \item[C2] Wait until the next living half-edge dies (spontaneously, due to the expiration of its exponential life-time). This is joined to 
  the one killed in previous step to form an edge of the graph along 
 which information has been transmitted. If the vertex it belongs to is sleeping, we change its status to active, along with all of its 
 half-edges. 
 Repeat from the first step. \label{c3}
 \end{enumerate}
 
 Every time C\ref{c1} is performed, we choose a vertex and trace the flow of influence from here onwards. Just before C\ref{c1} is 
 performed again, 
 when the number of active transmitter half-edges goes to $0$, we've explored the extent of the graph component that the chosen vertex 
 can influence, that had not been previously influenced.

In a  typical evolution of the exploration
process for large $n$ (number of nodes), 
the number of active transmitter half-edges visits  $0$
some number of times at an early stage of the exploration (these
times correspond to the completion of ``small'' influenced components)
before finally it takes off and stays strictly positive for a long period. The fist visit to~0
after this long period corresponds to the completion of a big
influenced component. In the ``fluid-limit''
scaling  of the process (when the number of nodes goes to infinity) 
 trajectories of the {\em fraction}  of active transmitter half-edges
converge to the deterministic  function~$H(e^{-t})$, where $H$ is
given by~(\ref{Hx}). The smallest strictly positive time $t_0>0$ for
which $H(e^{-t_0})=0$ approximates the time to the completion of a big
influenced component. Also, the fraction of all discovered nodes up to time
$t<t_0$ converges to $1 - G_D(e^{-t})$. 
It can be shown that the total size of all the  small components
discovered before 
the big one is negligible. Hence the fraction of nodes influenced in
the first  big component is approximately $1 - G_D(e^{-t_0})$, which 
is the first statement of Claim~\ref{c.alpha}. 
Uniqueness of such a big component can also be concluded form the fluid limit approximation.

\subsection{Analysis of the reverse-propagation process}
 One  introduces the following  dynamic to trace the possible sources of influence of a randomly chosen vertex.
 As in the forward process, we initially give all the half-edges
 i.i.d. random maximal lifetimes with exponential (mean one) distribution 
and then go through the following algorithm.

 \begin{enumerate}
  \item[D1] If there is no active half-edge (as in the beginning), select a sleeping vertex and declare it active, along with all its 
  half-edges.  For definiteness, we choose the vertex uniformly at random among all sleeping vertices. If there is no sleeping 
 vertex left, the process stops. \label{cb1}
  \item[D2] Pick an active half-edge and kill it. \label{cb2}
  \item[D3] Wait until the next transmitter half-edge dies (spontaneously). This is joined to the one killed in previous step to form an 
  edge of 
 the graph. 
 	If the vertex it belongs to is sleeping, we change its status to active, along with all of its half-edges. Repeat from the first 
 step. \label{cb3}
 \end{enumerate}
The analysis of this process
goes along the same lines as that of the forward one, with $H(x)$
and $G_D(x)$ replaced by $\overline H(x)$ and $G_{D^{(t)}}(x)$, respectively.

Additional work is needed (cf~\cite{viral} to formally relate the big component
of the dual process to the set of good pioneers. The branching approximation described in 
what follows confirms this formal approach.

\subsection{Branching-process approximation of the probability of
  choosing a good pioneer}
In this approach one approximates the exploration of the
original process at an early stage (before loops appear) 
 by a  Galton-Watson branching process, and conjectures that
the probability of the extinction of this process is equal to the
probability of choosing a good pioneer.
Using the well known result for the Galton-Watson branching process,
cf e.g.~\cite{Mass}, this probability can be shown equal to the
right-hand-side of~(\ref{e.alphabar}).

More precisely, if we start the exploration with a uniformly chosen
pioneer, its degree distribution follows $(D^{(r)},D^{(t)})$ with total degree, $D =D^{(r)}+D^{(t)}$. However,
since the probability of getting influenced is proportional to one's total degree, 
the degrees of the friends of this pioneer 
won't follow this  joint distribution. Their joint
receiver-and-transmitter degree distribution,
denoted by $(\widetilde{D}^{(r)},\widetilde{D}^{(t)})$, 
is given by
\begin{equation}
\P\{\,\widetilde{D}^{(r)}=v,\widetilde{D}^{(t)}=w\,\} =
\frac{\left(v+1\right){p}_{v+1,w}+\left(w+1\right){p}_{v,w+1}}{\E[D] }.
\end{equation}
The same (modified) distribution characterizes the nodes in the subsequent generations of the branching process.
The well known condition for the non-extinction of the branching process 
which diverges from the first-generation,
$$\mathbb{E}[{\widetilde{D}}^{(t)}]>1,$$
can be shown to agree with~(\ref{e.VA}). 
Further, if this condition is satisfied, the extinction probability
$\widetilde{p}_{ext}$ of
this  branching process
can be shown to be equal to  the smallest zero of $\overline H(x)$ (as defined
in~(\ref{Hbarx})) in the interval  $(0,1)$. Finally, the non-extinction probability
of the whole process started at the initial pioneer is equal
 $1-\mathbb{E}\left[({\widetilde{p}_{ext}})^{{D}^{\left(t\right)}}\right]$,
which agrees with the right-hand-side of~(\ref{e.alphabar}).


\end{document}